\documentclass[12pt]{article}
\usepackage{amsfonts,amssymb,latexsym,amsmath}
\usepackage{eucal}
\usepackage{geometry}
\usepackage{amsmath}
\usepackage{hyperref}

\usepackage{authblk}
\usepackage{etoolbox}
\usepackage{xcolor}
\usepackage{lipsum}
\usepackage{upgreek}
\usepackage[draft]{todonotes}

\newtheorem{thm}{Theorem}
\newtheorem{prop}{Proposition}
\newtheorem{coro}{Corollary}
\newtheorem{lemma}{Lemma}
\newtheorem{defn}{Definition}
\newtheorem{remark}{Remark}
\newtheorem{exmp}{Example}
\def\<{\langle}
\def\>{\rangle}

\newcommand\be{\begin{equation}} 
\newcommand\ee{\end{equation}}
\newcommand{\comment}[1]

\def\bea{\begin{eqnarray*} }

\def\eea{\end{eqnarray*} }

\begin{document}

\title{Reduced thin-sandwich equations on manifolds euclidean at infinity and on closed manifolds: existence and multiplicity}

\author{R. Avalos\thanks{Partially supported by PNPD/CAPES.}  \,\,and J. H. Lira\thanks
{Partially supported by CNPq and FUNCAP.}
}
\date{
\textit{Department of Mathematics, Federal University of Cear\'a,\\
Fortaleza, Brazil}
}
\maketitle

\begin{abstract} 
The reduced thin-sandwich equations (RTSE) appear within Wheeler's thin-sandwich approach towards the Einstein constraint equations (ECE) of general relativity. It is known that these equations cannot be well-posed in general, but, on closed manifolds, sufficient conditions for well-posedness have been established. In particular, it has been shown that the RTSE are well posed in a neighbourhood of umbilical solutions of the constraint equations without conformal Killing fields. In this paper we will analyse such set of equations on manifolds euclidean at infinity in a neighbourhood of asymptotically euclidean (AE) solutions of the ECE. The main conclusion in this direction is that on AE-manifolds admitting a Yamabe positive metric, the solutions of the RTSE parametrize an open subset in the space of solutions of the ECE. Also, we show that in the case of closed manifolds, these equations are well-posed around umbilical solutions of the ECE admitting Killing fields and present some relevant examples. Finally, it will be shown that in the set of umbilical solutions of the vacuum ECE on closed manifolds, the RTSE are generically well-posed. 
\end{abstract}

%






\setcounter{MaxMatrixCols}{10}



\section{Introduction}

The study of the Einstein field equations has been an active object of research in both mathematics and physics during the last century. In this line, plenty of the mathematical interest has been directed towards showing that these equations can be cast as a hyperbolic system, and thus, that there is a well-posed Cauchy problem associated with such a system. It has been shown that, at least for the most common sources, this is actually true as long as the initial data for this system satisfies some constraint equations. This last point has drawn a great deal of attention to the study of these constraint equations.  Since the Einstein equations are a system of second order equations, which can be cast as system of non-linear wave equations, the initial data for such system involves both the initial data for the 4-dimensional Lorentzian metric, which is given in terms of a Riemannian metric $g$ on a given 3-manifold $M$, and its initial time-derivative, which is given in terms of a symmetric second rank tensor field $\dot{g}$. In the case that other physical fields are involved in the field equations, the initial data set should take into account the necessary initial data for these field too. 

In order to simplify the discussion, for the moment we will restrict to vacuum, and thus neglect other possible fields besides the metric tensor. Typically, the constraint equations are not posed for $(g,\dot{g})$, but for $(g,K)$, where $K$ is a symmetric second-rank tensor field, which, after getting an embedding into space-time $(V,\bar{g})$, represents the extrinsic curvature of the Riemannian submanifold $(M,g)$. The relation between $K$ and $\dot{g}$ is very well-known, and it will be useful for us to write it down: 
\begin{align}\label{extcurv}
K=-\frac{1}{2N}(\dot{g}-\pounds_{\beta}g).
\end{align}
In the above expression, $N$ and $\beta$ represent respectively the usual lapse function and shift vector on the initial manifold $M$. Since the constraint system is best written in terms of $(g,K)$, and there is a very clear relation between $K$ and $\dot{g}$, solving the constraints for $(g,K)$ seems quite natural. Furthermore, from the perspective of the evolution problem, it is well known that $g$ and $K$ serve as the actual geometric initial data. This is because two initial data sets $(g,\dot{g}_1)$ and $(g,\dot{g}_2)$ which have the same data $(g,K)$ and solve the constraint equations, evolve into equivalent (diffeomorphic) space-times. This last remark shows that $(N,\beta)$ work actually as \textit{gauge} variables, with no further dynamical significance. This fact has naturally presented the idea of trying to solve the constraint equations for such gauge variables, and leave $g$ and $\dot{g}$ \textit{free}. This strategy would seem quite natural, and even desirable if one wanted to be able to specify arbitrarily the dynamical degrees of freedom. Furthermore, the constraint equations posed for $g$ and $K$ are a highly underdetermined  system, which presents the problem of \textit{choosing} which part of the initial data should be considered as free and which part should be determined by solving the constraints. Clearly, it would be desirable that both sets of initial data have a clear physical and geometric interpretation. The usual approach to this problem is to apply the so called \textit{conformal method} to rewrite the constraints as an elliptic system. In this approach, not $g$ neither $\dot{g}$ are left as freely specified initial data. This leaves open the question of whether it is actually possible to solve the constraints for the lapse and shift, and leave the dynamical variables $(g,\dot{g})$ as free data. In fact, this was put forward as a conjecture by John Wheeler, which coined the name \textit{thin sandwich conjecture} (TSC). 

In essence, the TSC states that, for freely given $(g,\dot{g})$ it is possible to solve the constraint equations for the lapse and shift. It has been shown that such program cannot work in general, since there are simple counterexamples \cite{B-O}. Furthermore, some global uniqueness results have also been shown \cite{B-O}-\cite{giulini1}, and sufficient conditions for the TSC to hold were first established in \cite{bartnik} for phenomenological sources (or vacuum) and extended in \cite{giulini1}, where some matter fields are included in a more fundamental level. All of these results were concerned almost exclusively with the 3-dimensional case, which is of interest for conventional general relativity. In \cite{ADRL}, the authors prove that the results presented in \cite{bartnik} actually hold in any dimension greater or equal to 3, and furthermore, that on any compact $n$-dimensional manifold there is an open subset in the space of solutions of the constraint equations where the TSC holds. This last point was a novelty also for the 3-dimensional case.  

Nonetheless, all of the above existence results are limited to the compact scenario. This is why we find it interesting to move to the non-compact setting, and try to see whether some of these results extend naturally or not. We will be interested in manifolds which have a finite number of ends, which are diffeomorphic to the complement of a ball in euclidean space. In this case it could be expected that, at least, a theorem of the type found in \cite{bartnik} should hold, since the main analytic techniques needed for such theorem are also valid (modulo technical details) in this non-compact scenario. Nevertheless, we are more interested in a result analogous to the one presented in \cite{ADRL}, which guarantees the existence of an open set in the space of solutions of the constraints where we can solve the equations in terms of the lapse and shift. The techniques used in \cite{ADRL} rely on a theorem proven by Kazdan and Warner \cite{KW}, which is no longer available in the non-compact setting. In order to prove an theorem analogous to the one presented in \cite{ADRL}, we will use some specific results valid for asymptotically euclidean manifolds \cite{CB1},\cite{maxwell-dilts},\cite{maxwell}. The main result coming from this analysis is the following:

\begin{thm}\label{thmTS2}
On any $n$-dimensional manifold euclidean at infinity, $n\geq 4$, which admits a Yamabe positive metric, the solutions of the reduced thin-sandwich equations parametrise an open subset in the space of solutions of the Einstein constraint equation. More explicitly, there is a 1-parameter family of umbilical solutions of the ECE which admits a neighbourhood that can be parametrised via the thin-sandwich problem. 
\end{thm}

We will then go back to the compact scenario, and investigate the possibility of solving the constraints in terms of the lapse-shift at some of those points where one of the conditions needed in the theorems presented in \cite{bartnik}-\cite{ADRL} fail. In this direction, we will show that using elliptic theory and a refinement of the functional spaces where the problem is posed, the main theorem presented in \cite{ADRL} can be extended and applied around solutions with Killing fields. Then, we will make use of the bifurcation analysis presented in \cite{crandal} to show that at some of these \textit{degenerate} initial data sets, we actually have bifurcation of solutions. Finally, we will show that the following generic result holds:

\begin{thm}
On any compact $n$-dimensional manifold, $n\geq 3$, the set of smooth \textit{umbilical} solutions of the vacuum constraint equations where the thin-sandwich problem is well posed is dense in the set of smooth umbilical solutions of the vacuum constraint equations.
\end{thm}

In fact, Theorem \ref{generic-umbilical} shows much more than the previous statement, since it is shown that any smooth umbilical solution can be approximated as much as we want, in any $C^k$-topology, by another umbilical solution, with the same mean and scalar curvatures as the original one, for which the thin-sandwich problem is well-posed.

\section{Lapse-shift formulation of the constraint equations}

The constraint equations of general relativity on an $n$-dimensional manifold $M$ read as follows:
\begin{align}
R_g+(\text{tr}_{g}K)^2-|K|^2_g - 2 \Lambda&=2\epsilon\label{hamiltonian},\\
\nabla\text{tr}_gK - \text{div}_g K &= S \label{momentum},
\end{align}
where $g$ is a Riemannian metric on $M$, $R_g$ is its scalar curvature, $\epsilon$ represents the induced energy density on $M$, $S\doteq NT_{0i}$ represents minus the induced momentum density on $M$ and $\Lambda$ the cosmological constant. The tensor $K$ is a symmetric second rank tensor field which is related to $\dot{g}$ as in (\ref{extcurv}), by means of the lapse function and shift vector.

From now on we will only take into consideration globally hyperbolic space-times of the form $(\bar M=M\times\mathbb{R},\bar{g})$. It will be convenient to consider the adapted frame 
\begin{align*}
e_{i}  &  = \partial_{i}\; ,\; \; i=1,\cdots,n,\\
e_{0}  &  = \partial_{t} - \beta
\end{align*}
and its dual coframe
\begin{align*}
\theta^{i}  &  = dx^{i} +\beta^{i}dt \; , \; \; i=1,\cdots,n,\\
\theta^{0}  &  = dt,
\end{align*}
where $\{x^i\}_{i=1}^{n}$ are local coordinates on $M$ and $t$ is a coordinate on $\mathbb{R}$. Denoting respectively by $N$ and $\beta$ the lapse function and the shift vector field, the ambient metric $\bar g$ is written in terms of the coframe $\{\theta^a\}_{a=0}^n$ in the following way
\begin{align*}
\bar g = -N^{2}\theta^0\otimes\theta^0+ g_{ij}\theta^{i}\otimes\theta^{j}%
\end{align*}
As in \cite{bartnik} and \cite{ADRL}, we consider $\epsilon$ and $S$ as phenomenologically given (or zero). For any solution of the constraints satisfying $2\epsilon_{\Lambda}-R_g>0$, where $\epsilon_{\Lambda}=\epsilon+\Lambda$, it follows from (\ref{hamiltonian}) that
\begin{equation}
\label{lapse}N=\sqrt{\frac{(\mathrm{tr}_{g} \gamma)^{2}-|\gamma|^{2}_{g}%
}{2\epsilon_{\Lambda}-R_{g}}}%
\end{equation}
where the tensor $\gamma$ has components
\begin{equation}
\gamma_{ij}=\frac{1}{2}\big(\dot g_{ij}-(\nabla_{i}\beta_{j}+\nabla_{j}%
\beta_{i})\big).
\end{equation}
The expression (\ref{lapse}) defines the lapse function $N$ in terms of $g,\dot{g}$ and $\beta$. Using this information, we see that (\ref{momentum}) is equivalent to
\begin{equation}\label{div-S}
\Phi(\psi,\beta)\doteq \mathrm{div}_g \Bigg(\sqrt{\frac{2\epsilon_{\Lambda}-R_{g}}{(\mathrm{tr}_{g}
\gamma)^{2}-|\gamma|^{2}_{g}}}\,\big(\gamma-\mathrm{tr}_{g} \gamma\,
g\big)\Bigg) - S = 0.
\end{equation}
where
\begin{equation}
\psi \doteq (g,\dot{g},\Lambda,\epsilon,S).
\end{equation}
The strategy adopted in both \cite{bartnik} and \cite{ADRL} is to reverse the above procedure. That is, to consider (\ref{div-S}) as an equation for the shift vector field $\beta$, when $\psi$ is a freely chosen initial data. If we can show that (\ref{div-S}) is well posed in suitable functional spaces, then using (\ref{lapse}) as a definition for the lapse function, we get a solution for the constraint equations for the prescribed data $\psi$. It is clear that (\ref{div-S}) is not well-defined for arbitrarily chosen $\psi$. Thus, the strategy is to analyse such system in a neighbourhood of a solution of the constraint equations inducing data $(\psi_0,\beta_0)$ such that $\Phi(\psi_0,\beta_0)$ is well-defined. The natural way to address this problem is to analyse the linearised map $D_2\Phi_{(\phi_0,\beta_0)}$ and try to apply an implicit function argument.  

Before delving into the technical core of the paper, it will be worth noticing that there is a simple way to incorporate non-phenomenological sources into the picture described above. For instance, in the case of sources generated by a scalar field and an electromagnetic field, which are the most common classical fields, we set
\begin{align*}
\epsilon_{\Lambda}&=\Lambda+\frac{1}{2}(\pi^2+|\nabla\phi|^2_{g})+V(\phi)+\frac{1}{2}(|E|^2_g+|B|^2_g),\\
S&=\pi\nabla\phi+F(\cdot,E),
\end{align*}
where $\phi$ represents a real scalar field with a self-interacting potential $V$ and $\pi\doteq \frac{1}{N}e_0(\phi)|_{t=0}$. The two-form  $F$ on $M$ is induced from the 2-form $\bar{F}$ in $\bar M$ which describes the full electromagnetic field in space-time. Hence,  $E$ and $B$ represent, respectively, the prescribed electric and magnetic fields in $M$, defined by $E^i\doteq \frac{1}{N}\bar{F}_0{}^{i}|_{t=0}$ and $|B|^2\doteq \frac{1}{2}F^{ij}F_{ij}|_{t=0}$. At this point we should decide what could actually be a sensible choice for the specified initial data. For instance, from the physical point of view, it is reasonable to consider $E$ and $B$  as given, describing some configuration of electric charges and currents. In this case, we define
\[
\rho\doteq \Lambda + \frac{1}{2}|\nabla\phi|^2_{g}+V(\phi)+\frac{1}{2}(|E|^2_g+|B|^2_g)
\]
Thus, we proceed as above to show that, given a solution of the constraint equations satisfying $2\rho-R_g>0$,  the lapse can be determined from the hamiltonian constraint by the expression
\begin{align*}
N=\sqrt{\frac{(\mathrm{tr}_{g} \gamma)^{2}-|\gamma|^{2}_{g}-(\dot{\phi}-d\phi(\beta))^2}{2\rho-R_{g}}},
\end{align*}
where now $\dot{\phi}\doteq \partial_t\bar{\phi}|_{t=0}$ is part of the initial data for evolution system. This implies that the momentum constraint is equivalent to
\begin{align}\label{redeq2}
\begin{split}
&\mathrm{div}_g\left(\sqrt{\frac{2\rho-R_{g}}{(\mathrm{tr}_{g} \gamma)^{2}-|\gamma|^{2}_{g}-(\dot{\phi}-d\phi(\beta))^2}}(\gamma-g\mathrm{tr}_g\gamma)\right)-\\
&\sqrt{\frac{2\rho-R_{g}}{(\mathrm{tr}_{g} \gamma)^{2}-|\gamma|^{2}_{g}-(\dot{\phi}-d\phi(\beta))^2}}\,(\dot{\phi}-d\phi(\beta))d\phi-F(\cdot,E)=0.
\end{split}
\end{align}
The procedure to be followed in this case should be the same as described previously, but considering the above operator instead of the one defined in (\ref{div-S}), and adding the necessary initial data $(\phi,\dot{\phi},E,B)$. Throughout this paper we will use both approaches, treating sometimes the sources either as phenomenological or as fundamental data, what corresponds respectively to the formulations (\ref{div-S}) and (\ref{redeq2}). We will use $\Phi$ to denote the non-linear operator defined in both situations, and we will refer to the resulting system of partial differential equations as the reduced thin sandwich equations (RTSE). Now, since the electromagnetic field produces an extra constraint given by $\mathrm{div}_gE=0$, this equation should be coupled to the above system. However, in the momentum constraint for the shift (in a neighborhood of an appropriately chosen reference solution of the constraint equations), $g$ is treated as a given datum. Hence, the electromagnetic constraint would decouple from the momentum constraint. Indeed, we can suppose that the $E$ has been chosen such that this electromagnetic constraint is satisfied.

In sum, we can affirm that the energy and momentum densities are generated by more fundamental physical fields, namely $\Lambda,\phi,E, F$, which can be treated as independent of the lapse and shift. 

Our discussion will be considerably simplified by the fact that we will be interested in studying the well-posedness of the above system around a vacuum solution with cosmological constant. Indeed, we will consider initial data satisfying $K=\alpha g$, with $\alpha$ a non-zero constant. In this way, the momentum constraint is automatically satisfied, and the hamiltonian constraint reads as follows:
\begin{align}\label{prescribedcurv}
R_g+\alpha^2n(n-1)=2\Lambda.
\end{align}
For solutions such that $g$ is asymptotically euclidean in a sense to be made precise later, the scalar curvature $R_g$ must approach zero as we go to infinity and thus we have to pick $\alpha^2=\frac{2\Lambda}{n(n-1)}$. This means that we must consider a positive cosmological constant. This choice of a cosmological constant implies that  we have a initial configuration which evolves into an space-time which is either expanding or contracting as a whole.
We can produce easy examples starting with 
 an asymptotically flat metric which is scalar flat and corresponds to   
 totally geodesic vacuum solutions of the constraint equations without cosmological constant. Such a solution would have $K=0$ and satisfy $R_g=0$. Thus, if we consider $\alpha^2=\frac{2\Lambda}{n(n-1)}$ and $K'=\alpha g$, we get that $(g,K')$ solves the constraint equations with a positive cosmological constant $\Lambda$. We could get trivial examples from Schwarzschild's initial data or Minkowski's initial data. This is enough to convince ourselves that there are meaningful examples satisfying our requirements. Nonetheless, we will later show that on any manifold euclidean at infinity admitting a Yamabe positive metric there are initial data sets of the type we are considering.

\section{Manifolds euclidean at infinity and elliptic systems}


We now present the basic definition and results concerning elliptic systems defined on manifolds euclidean at infinity that will be used throughout this paper. The details can be found in \cite{CB1}.

\begin{defn}
A $n$-dimensional smooth Riemannian manifold $(M,e)$ is called euclidean at infinity if there exits a compact set $B\subset M$ such that $M\backslash B$ is the disjoint union of a finite number of open sets $U_i$, such that each $(U_i,e)$ is isometric to the exterior of an open ball in Euclidean space.   
\end{defn}

On manifolds euclidean at infinity we define $d=d(x,p)$ the distance in the Riemannian metric $e$ of an arbitrary point $x$ to a fixed point $p$. We will typically omit the dependence on $(x,p)$. 

\begin{defn} Let $(M,e)$ be a manifold euclidean at infinity. A weighted Sobolev space $H_{s,\delta}$, with $s$ a nonnegative integer and $\delta\in\mathbb{R}$ is the space of tensor fields of some given type {\rm (}functions, for instance{\rm)} on $(M,e)$ with generalized derivatives of order up to $s$ in the metric $e$ such that $(1 + d^2)^{\frac{1}{2}(m+\delta)}D^mu\in L^2$, for $0\leq m\leq s$. It is a Banach space with norm
\begin{align}\label{norm}
||u||^{2}_{H_{s,\delta}}\doteq \sum_{m=0}^s\int_M|D^mu|_e^2(1 + d^2)^{(m+\delta)}\mu_{e}
\end{align}
where $D$ represents the $e$-covariant derivative and $\mu_e$ the Riemannian volume form associated with $e$.
\end{defn}

\begin{defn}
We will say that a Riemannian metric $g$ on $M$ is asymptotically euclidean (AE) if $g-e\in H_{s,\delta}$, $s>\frac{n}{2}$ and $\delta>-\frac{n}{2}$.
\end{defn}

These weighted spaces share several of the properties of the usual Sobolev spaces. For instance, if we consider the following weighted norm in the space of $C^k$-tensor fields of some given type
\begin{align*}
||u||_{C^k_{\delta}}=\sup_{x\in M}\sum_{m=0}^k(1+d^2)^{\frac{1}{2}(\delta+m)}|D^m u|_e
\end{align*}
we have a continuous embedding stated as follows:

\begin{lemma}\label{sobolevembedding}
Let $(M, e)$ be a manifold euclidean at infinity.  If $s' <s -n/2$ and $\delta' <\delta +n/2$, the inclusion
\begin{align*}
H_{s,\delta}\subset C^{s'}_{\delta'}
\end{align*}
holds and it is continuous.
\end{lemma}

Also, the following multiplication property is very useful to analyse the range of differential opertaros acting between these weighted spaces. 

\begin{lemma}
If $(M,e)$ is a manifold euclidean at infinity, then the following continuous multiplication property holds
\begin{align*}
H_{s_1,\delta_2}\times H_{s_2,\delta_2}&\mapsto H_{s,\delta},\\
(f_1,f_2)&\mapsto f_1\otimes f_2,
\end{align*}
if $s_1,s_2\geq s$, $s<s_1+s_2-\frac{n}{2}$ and $\delta<\delta_1+\delta_2+\frac{n}{2}$.
\end{lemma}

We will need to analyse linear elliptic operators in this context. For us, the following setting will be general enough. Consider a linear elliptic operator of the following from:
\begin{align*}
L=\sum_{m=0}^{N} a_m D^m
\end{align*}
acting between sections of some tensor bundles $E$ and $F$ over an asymptotically euclidean manifold $(M,e)$. Suppose that the coefficients $a_m$, $0\leq m\leq N$, are sections of the tensor bundles $(\otimes TM)^m\otimes E^{*}\otimes F$ and satisfy the following regularity conditions:
\begin{align*}
&1)\, a_m\in H_{s_m,\delta_m}, 0\leq m\leq N-1,\\
&2)\, a_N-A \in H_{s_N,\delta_N},
\end{align*}
where $s_m>\frac{n}{2}+m-N+1$ and $\delta_m>N-m-\frac{n}{2}$ for $0\leq m\leq N$, and $A$ is a smooth tensor field on $M$ which is constant in each end, and such that $L_{\infty}\doteq A\cdot D^N$ is elliptic. Under these assumptions we have that, if $s_m\geq s-N\geq 0$, then $L$ defines a continuous map from $H_{s,\delta}$ to $H_{s-N,\delta+N}$, for any  $\delta\in \mathbb{R}$. Furthermore, the following theorem holds (see \cite{CB1}):

\begin{thm}\label{elliptic}
Under the above hypotheses, if $s\geq N$, $s_m\geq s-N$ and $-\frac{n}{2}< \delta< \frac{n}{2}-N$, then $L$ maps $H_{s,\delta}$ into $H_{s-N,\delta+N}$ with finite dimensional kernel and closed range. 
\end{thm}

We will need a few more results concerning weighted Sobolev spaces on manifolds euclidean at infinity. The following lemma is crucial to justify some integration by parts procedures typically performed when analysing the kernel of an elliptic operator of the kind described above. For the proof see, for instance, \cite{CB2}. From now on, we will only consider linear operators $L$ of second order.

\begin{lemma}\label{integrationbyparts} Suppose that the coefficients of a second order linear opeator $L$ satisfy $a_2-A\in H_{s,\delta_q}$, $a_1\in H_{s-1,\delta_q+1}$ and $a_0\in H_{s-2,\delta_q+2}$, with $s>\frac{n}{2}$, $s\geq 2$, $\delta_q>-\frac{n}{2}$. 
Given $f\in H_{s-2,\tilde{\delta}+2}$ with $\tilde{\delta} \ge\delta_q$, let $u\in H_{s,\delta_q}$ be a solution of $Lu=f$.
Then, $u\in H_{s,\tilde{\delta}}$ if $\tilde{\delta}<\frac{n}{2}-2$.
\end{lemma}

Finally, we will use the following lemma, which establishes a weighted Poincar\'e inequality for manifolds euclidean at infinity. Such inequality was proved by Robert Bartnik in $\mathbb{R}^n$ in \cite{bartnik2}, and the more general proof can be found in \cite{maxwell-dilts}.

\begin{lemma}\label{poincare}
Let $(M,g)$ be an $n$-dimensional manifold, $n\geq 3$, which is $H_{2,\delta}$-AE, with $\delta>-\frac{n}{2}$. Then, there exists a constant depending on $n$ such that for any function $u\in H_{1,-1}$ the following inequality holds
\begin{align}
||u||_{L^2_{-1}}\leq C||Du||_{L^2}.
\end{align} 
\end{lemma}

\section{Sufficient conditions for the existence of solutions on AE manifolds}

In this section we will establish conditions for the well-posedness of the system (\ref{redeq2}) on a manifold euclidean at infinity $(M,e)$, around a reference solution $(\psi_0,\beta_0)$, where $\psi_0=(g_0,\dot{g}_0,\Lambda,\epsilon_0,S_0)$. 
We set $\epsilon_0=0$, $S_0=0$ and $\Lambda>0$. We also assume that $g_0$ is $H_{s+3,\delta+1}$-asymptotically flat with $R_{g_0}=0$, for some $s>\frac{n}{2}$ and $\delta>-\frac{n}{2}$. Finally, we fix $K_0=\alpha g_0$, with $\alpha^2=\frac{2\Lambda}{n(n-1)}$.
As discussed above these choices would provide a solution of the constraint equations. Furthermore, from this solution, by choosing a strictly positive and bounded lapse function $N_0$ satisfying $N_0-1\in H_{s+3,\delta+1}$ and a shift vector $\beta_0\in H_{s+2,\delta}$, we get
\begin{align}
\dot{g}_0=-2\alpha N_0 g_0 + \pounds_{\beta_0}g_0,
\end{align}
It will be useful to note that, having metrics $g$ and $e$ defined on $M$, their covariant derivatives will be related by means of a tensor field $S$, defined in coordinates by
\begin{align}
\begin{split}
S^i_{jk}&\doteq{}^{g}\,\!\Gamma^i_{kj}-{}^{e}\,\!\Gamma^{i}_{jk}=\frac{g^{iu}}{2}(D_jg_{uk}+D_kg_{uj}-D_{u}g_{jk}),
\end{split}
\end{align}
where ${}^{g}\,\!\Gamma^i_{kj}$ and ${}^{e}\,\!\Gamma^{i}_{jk}$ denote the connection coefficients related with $g$ and $e$ respectively. Using the above expression, the fact that $g-e\in H_{s+3,\delta+1}$ and the multiplication property, we see that $S\in H_{s+2,\delta+2}$. Using this fact and the multiplication property, we get that differentiation with respect to $\nabla$ takes $H_{s,\delta}$ to $H_{s-1,\delta+1}$. Thus, we find that $\dot{g}_0+2\alpha e\in H_{s+1,\delta+1}$. 

Since we are restricted to a neighbourhood of the fixed reference solution, we only consider metrics $g$ which are $H_{s+3,\delta+1}$-asymptotically flat and such that $||g-g_0||_{H_{s+3,\delta+1}}$ is small enough. In a similar way, we will consider $\dot{g}$ with the same asymptotic behaviour as $\dot{g}_0$, that is, $\dot{g}+2\alpha e\in H_{s+1,\delta+1}$, such that $||\dot{g}-\dot{g}_0||_{H_{s+1,\delta+1}}$ is sufficiently small. In this way it will be useful for us to rewrite $\dot{g}=-2\alpha e + \delta \dot{g}$
and take that $\delta\dot{g}\in H_{s+1,\delta+1}$ as part of the data that is actually freely specified. Since the same can be done with $g$, we will also take $\delta g\doteq g-e\in H_{s+3,\delta+1}$ as a freely given datum. In this way, it is clear that $\gamma$ is a continuous function of $\delta g$, $\delta\dot{g}$ and $\beta$, using the previously discussed functional spaces, with $\gamma+\alpha e\in H_{s+1,\delta+1}$. Furthermore, we get that $\mathrm{tr}_g\gamma+n\alpha\in H_{s+1,\delta+1}$ and $|\gamma|_g^2-n\alpha^2 \in H_{s+1,\delta+1}$. For the remaining initial data, we assume that  $\phi\in H_{s+2,\delta+2}$, $\dot{\phi}\in H_{s+1,\delta+2}$, $E\in H_{s+1,\delta+2}$ and $F\in H_{s+1,\delta+2}$. Finally, we suppose that the potential function $V$ is given by a continuous map from $H_{s+2,\delta+2}$ to $H_{s+2,\delta+3}$. Fixed this setting, one gets that  
\[
N(\delta g,\phi,E,F,\delta\dot{g},\dot{\phi},\beta)\neq 0
\]
in a small enough neighbourhood of the reference solution. Hence, $N^{-1}$ is well defined in such a neighborhood, and furthermore, $N^{-1}-1\in H_{s+1,\delta+1}$ on it. 

For simplicity, we denote  $\mathcal{E}_1\doteq H_{s+3,\delta+1}\times H_{s+2,\delta+2}\times H_{s+1,\delta+2}\times H_{s+1,\delta+2}\times H_{s+1,\delta+1}\times H_{s+1,\delta+2}$, $\mathcal{E}_2\doteq H_{s+2,\delta}$ and $\mathcal{Z}\doteq H_{s,\delta+2}$.  The freely specified initial data are denoted by  $\psi\doteq (\delta g,\phi,E,F,\delta\dot{g},\dot{\phi})\in \mathcal{E}_1$. Hence,  we pick up a small enough neighbourhood $U$ of the initial data for the reference solution $(\delta g_0,0,0,0,\delta\dot{g}_0,0,\beta_0)$ such that the map $\Phi$ given by (\ref{redeq2}) can be written as a $C^1$-map of the form
\begin{align}
\centering
\begin{split}
 \Phi: U\subset \mathcal{E}_1 \times \mathcal{E}_2 \to \mathcal{Z}\\
(\psi,\beta)  =(\delta g,\phi,E,F,\delta\dot{g},\dot{\phi},\beta)&\mapsto \Phi(\delta g,\phi,E,F,\delta\dot{g},\dot{\phi},\beta),
\end{split}
\end{align}
Now, notice that the derivative of $\Phi$ with respect to the variable $\beta\in \mathcal{E}_2$ at a point $(\psi,\beta)$ is given by 
\begin{align}\label{linearisation.0}
\begin{split}
D_2\Phi_{(\psi,\beta)}\cdot Y=&\mathrm{div}_g\left( \frac{1}{N}\left( g\mathrm{div}_gY - \frac{1}{2}\pounds_{Y}g - \frac{1}{2\rho - R_g}\left( \langle P_g,\nabla Y\rangle_g - \pi d\phi(Y) \right)P_g \right) \right)\\
&-\frac{1}{N}\left( \frac{1}{2\rho - R_g}\left( \langle P_g,\nabla Y\rangle_g - \pi d\phi(Y) \right)\pi - d\phi(Y) \right)d\phi,
\end{split}
\end{align}
where
\begin{align*}
\pi&\doteq N^{-1}(\dot{\phi} - d\phi(\beta)),\\
P_g&\doteq K - g\mathrm{tr}_gK,
\end{align*}
and the rest of the quantities have already been defined. From this, we get the following computational lemma.
\begin{lemma}
The linearisation of the map $\Phi:U\subset\mathcal{E}_1\times \mathcal{E}_2\mapsto \mathcal{Z}$ around a vacuum solution $(g_0,K_0)$ of the constraint equations, with $K_0=\alpha g_0$ for some constant $\alpha>0$, producing $\psi_0=(\delta g_0,0,0,0,\delta\dot{g}_0,0,\beta_0)\in \mathcal{E}_1\times\mathcal{E}_2$, is given by 
\begin{align}
D_2\Phi_{(\psi_0,\beta_0)}\cdot Y=-\frac{1}{2}\mathrm{div}_g\left(\frac{1}{N_0}\pounds_{g_0,conf}Y \right),
\end{align}
where
\begin{align*}
\pounds_{g_0,conf}Y\doteq \pounds_Yg_0 - \frac{2}{n}g_0\mathrm{div}_{g_0}Y.
\end{align*}
\end{lemma}
\begin{proof}
First, notice that under our choices, we get that $\dot{g}_0$ is defined by the relation
\begin{align*}
K_0=-\frac{1}{2N_0}\left(\dot{g}_0 - \pounds_{\beta_0}g_0 \right),
\end{align*}
where $N_0$ is a sufficiently smooth strictly positive bounded function satisfying $N_{0}-1\in H_{s+3,\delta+1}$, and $\beta_0\in H_{s+2,\delta}$ is chosen freely. In this situation, we write $\dot{g}_0\doteq -2\alpha N_0 e + \delta\dot{g}_0$, for some $\delta\dot{g}_0$, which is fixed using the above relation by
\begin{align*}
\delta\dot{g}_0\doteq - 2\alpha (N_0-1)e - 2\alpha(g_0-e) - 2\alpha(N_0-1)(g_0-e) + \pounds_{\beta_0}g_0 \in H_{s+1,\delta+1}.
\end{align*}
With this choice we have fixed $\psi_0$ and $\beta_0$ and we can use these choices in (\ref{linearisation.0}) to compute $D_2\Phi_{(\psi_0,\beta_0)}$. With this in mind, first notice that $P_{g_0}=\alpha(1 -n)g_0$, $\pi_0=\phi_0=0$ and, since $(g_0,K_0)$ solve the vacuum constraints, $2\rho_0 - R_{g_0}=2\Lambda - R_{g_0}= (\mathrm{tr}_{g_0}K_0)^2 - |K_0|^2_{g_0}=\alpha^2n(n - 1)$. Thus
\begin{align*}
D_2\Phi_{(\psi_0,\beta_0)}\cdot Y &=\mathrm{div}_{g_0}\left( \frac{1}{N_0}\left( g_0 \mathrm{div}_{g_0}Y - \frac{1}{2}\pounds_{Y}g_0 - \frac{1}{\alpha^2n(n - 1)} \alpha^2(1-n)^2 \langle g_0,\nabla_{g_0} Y\rangle_{g_0} g_0 \right)\right)\\
&=\mathrm{div}_{g_0}\left( \frac{1}{N_0}\left( g_0 \mathrm{div}_{g_0}Y - \frac{1}{2}\pounds_{Y}g_0 - \frac{n-1}{n}  \mathrm{div}_{g_0}Y g_0 \right)\right),\\
&=- \frac{1}{2}\mathrm{div}_{g_0}\left( \frac{1}{N_0}\left(  \pounds_{Y}g_0 - \frac{2}{n}  \mathrm{div}_{g_0}Y g_0 \right)\right),
\end{align*}
which finishes the proof.
\end{proof}
\begin{coro}
The linearisation of the map $\Phi$ defined above, taken at an umbilical solution of the vacuum constraint equations with positive consmological constant $\Lambda$, defines an elliptic operator $D_2\Phi_{(\psi_0,\beta_0)}:\mathcal{E}_2\mapsto\mathcal{Z}$. 
\end{coro}
\begin{proof}
From the above lemma, we know that the principal part of $D_2\Phi_{(\psi_0,\beta_0)}$ is given by $-\frac{1}{2N_0}\Delta_{g_0,conf}:\mathcal{E}_2\mapsto\mathcal{Z}$, where  $\Delta_{g_0,conf}Y\doteq \mathrm{div}_{g_0}\left(\pounds_{g_0,conf}Y\right)$. This operator is known to be elliptic, from which we get that so is $D_2\Phi_{(\psi_0,\beta_0)}$.
\end{proof}
\bigskip

In the next lemma we will characterise the kernel of $L\doteq D_2\Phi_{(\psi_0,\beta_0)}$, where the linearisation is taken at an umbilical reference solution of the constraint equations. It is clear the the set of conformal Killing fields of $g_0$ lies in $\mathrm{Ker}(L)$. We will now show that actually $\mathrm{Ker}(L)=\mathrm{Ker}(\Delta_{g_0,conf})$. That is, $\mathrm{Ker}(L)$ coincides with the space of conformal Killing fields of $g_0$.

\begin{lemma}\label{injectivity}
Consider the linearisation of the map $\Phi$ defined above taken at a point $(\psi_0,\beta_0)$ constructed from an umbilical solution of the vacuum constraint equations with positive cosmological constant $\Lambda$, together with choices of lapse and shift satisfying $\beta_0\in \mathcal{E}_2$ and $N_0$ a strictly positive bounded function such that $N_0-1\in H_{s+3,\delta+1}$. Denote this linear map by $L:\mathcal{E}_2\mapsto \mathcal{Z}$. Then, the kernel of $L$ is the space of conformal Killing field of $g_0$.
\end{lemma}
\begin{proof}
Consider $Y\in \mathrm{Ker}(L)$. Our first claim is that for $n\geq 3$, it actually holds that $Y\in H_{s+2,-1}$. This is consequence of Lemma \ref{integrationbyparts}, since we can rewrite $L$ in the following way
\begin{align*}
L(Y)_k=\tilde{a}_2{}^{jiu}_k \nabla_j\nabla_i Y_u + \tilde{a}_1{}^{iu}_k\nabla_i Y_u, 
\end{align*}
where
\begin{align*}
{\tilde{a}_2}{}^{jiu}_k&=-\frac{1}{2N_0}\left(g_0^{ji}\delta^u_k + g_0^{ju}\delta^i_k - \frac{2}{n}g_0^{iu}\delta^j_k \right),\\
\tilde{a}_1{}^{iu}_k&=\frac{1}{2N_0^2}\left(g^{ji}\delta^u_k + g^{ju}\delta^i_k - \frac{2}{n}g^{iu}\delta^j_k  \right)\nabla_jN_0=-\frac{1}{N_0}\tilde{a}_2{}^{jiu}_k\nabla_jN_0.
\end{align*}
Also, recalling that $D$ denotes the $e$-covariant derivative, we have that
\begin{align*}
\nabla_iY_u&=D_iY_u - S^l_{iu}Y_l,\\
\nabla_j\nabla_iY_u
&=D_jD_iY_u - \left( S^a_{iu}\delta^b_j + S^b_{ji}\delta^a_u  + S^a_{ju}\delta^b_i\right)D_bY_a + \left(S^a_{ji}S^l_{au} + S^a_{ju}S^l_{ia} - D_jS^l_{iu}\right)Y_l.
\end{align*}
Thus, we get
\begin{align}\label{ellipticsystem1}
L(Y)_k&= a_2{}^{lba}_kD_lD_bY_a + a_1{}^{ba}_k D_bY_a + a_0{}^{a}_k Y_a, 
\end{align}
where
\begin{align*}
a_2{}^{lba}_k&=\tilde{a}_2{}^{lba}_k,\\
a_1{}^{ba}_k&= \tilde{a}_1{}^{ba}_k - \tilde{a}_2{}^{jiu}_k\left( S^a_{iu}\delta^b_j + S^b_{ji}\delta^a_u  + S^a_{ju}\delta^b_i\right),\\
a_0{}^{a}_k&= \tilde{a}_2{}^{jiu}_k\left(S^b_{ji}S^a_{bu} + S^b_{ju}S^a_{ib} - D_jS^a_{iu}\right)   - \tilde{a}_1{}^{iu}_kS^a_{iu}.
\end{align*}
Now, consider the model operator $L_{\infty}Y_k\doteq A_2{}^{lba}_{k}D_lD_bY_a$, where
\begin{align*}
A_2{}^{jiu}_{k}=-\frac{1}{2}\left(e^{ji}\delta^u_k + e^{ju}\delta^i_k - \frac{2}{n}e^{iu}\delta^j_k \right).
\end{align*}
This operator actually defines the operator $L_{\infty}=-\frac{1}{2}\Delta_{e,conf}$, which is clearly an elliptic operator, with $A_2$ smooth and is constant in each end. Furthermore, some computations and the use of the multiplication property show that, under our functional hypotheses, it holds that  
$a_2-A_2\in H_{s+3,\delta+1}$, for $s>\frac{n}{2}$ and $\delta>-\frac{n}{2}$. This implies that we can rewrite
\begin{align*}
\tilde{a}_1{}^{iu}_k
&= -\left(\frac{1}{N_0}-1\right)(\tilde{a}_2{}^{jiu}_k - A_2{}^{jiu}_k)D_jN_0 - (\tilde{a}_2{}^{jiu}_k - A_2{}^{jiu}_k)D_jN_0 - \left(\frac{1}{N_0}-1\right)\tilde{A}_2{}^{jiu}_kD_jN_0 \\
&- \tilde{A}_2{}^{jiu}_kD_jN_0.
\end{align*}
Now, since by hypotheses $N_0-1\in H_{s+3,\delta+1}$, then $DN\in H_{s+2,\delta+2}$. This implies that $(N^{-1}_0-1)DN_0 \in H_{s+2,\delta+3},(\tilde{a}_2-A_2)\otimes DN\in H_{s+2,\delta+3}$ and $(N^{-1}_0-1)(\tilde{a}_2-A_2)\otimes DN\in H_{s+2,\delta+3}$ from the multiplication property, which shows that $\tilde{a}_1\in H_{s+2,\delta+2}$. With this in mind, we can see that
\begin{align*}
a_1{}^{ba}_k&= \tilde{a}_1{}^{ba}_k - (\tilde{a}_2{}^{jiu}_k-A_2{}^{jiu}_k)\left( S^a_{iu}\delta^b_j + S^b_{ji}\delta^a_u  + S^a_{ju}\delta^b_i\right) - A_2{}^{jiu}_k\left( S^a_{iu}\delta^b_j + S^b_{ji}\delta^a_u  + S^a_{ju}\delta^b_i\right).
\end{align*}
From the hypothesis $g-e\in H_{s+3,\delta+1}$, we get $S\in H_{s+2,\delta+2}$, which implies that the third term is in this same space; from the multiplication property we see that the second one is also in $H_{s+2,\delta+2}$ , which, since $\tilde{a}_1\in H_{s+2,\delta+2}$, implies that $a_1\in H_{s+2,\delta+2}$. Similarly, we get
\begin{align*}
a_0{}^{a}_k&= (\tilde{a}_2{}^{jiu}_k - A_2{}^{jiu}_k)\left(S^b_{ji}S^a_{bu} + S^b_{ju}S^a_{ib} - D_jS^a_{iu}\right) + A_2{}^{jiu}_k\left(S^b_{ji}S^a_{bu} + S^b_{ju}S^a_{ib} - D_jS^a_{iu}\right)   \\
&- \tilde{a}_1{}^{iu}_kS^a_{iu}.
\end{align*}
Then, we get that $S\otimes S \in H_{s+2,\delta+4}$; $S\otimes \tilde{a}_1\in H_{s+2,\delta+4}$; also $DS\in H_{s+1,\delta+3}$, implying that the first and second terms are in $H_{s+1,\delta+3}$. This implies that $a_0\in H_{s+1,\delta+3}$. These computations imply that the coefficients in (\ref{ellipticsystem1}) satisfy the hypotheses of Lemma \ref{integrationbyparts}. This shows that for $L:H_{s+2,\delta}\mapsto H_{s,\delta+2}$, if $Y\in \mathrm{Ker}(L)$, then $Y\in H_{s+2,\tilde{\delta}}$ for any $\tilde{\delta}<\frac{n}{2}-2$, which for $n\geq 3$ implies $Y\in H_{s+2,-1}$. 

All of the above was done so that we can justify the following integration by parts formula for any $Y\in \mathrm{Ker}(L)$. Notice that for $Y\in C^{\infty}_0$, it holds that
\begin{align*}
\int_M\langle \mathrm{div}_{g_0}\left(\frac{1}{N_0}\pounds_{g_0,conf}Y \right),Y\rangle_{g_0}\mu_{g_0}&= - \int_M\frac{1}{N_0}\langle \pounds_{g_0,conf}Y ,\nabla Y\rangle_{g_0}\mu_{g_0},
\end{align*}
Now, consider $Y\in \mathrm{Ker}(L)\subset H_{s+2,-1}$, and $\{Y_n\}_{n=1}^{\infty}\subset C^{\infty}_{0}$, such that $Y_n\xrightarrow[n\mapsto\infty]{H_{s+2,-1}} Y$, and notice that
\begin{align*}
\Big|\int_M\frac{1}{N_0}\langle \pounds_{g_0,conf}Y ,\nabla Y\rangle_{g_0}\mu_{g_0} - \int_M\frac{1}{N_0}\langle \pounds_{g_0,conf}Y_n ,\nabla Y_n\rangle_{g_0}\mu_{g_0}\Big|&\leq \int_M \frac{1}{N_0}|\langle \pounds_{g_0,conf}(Y-Y_n) ,\nabla Y\rangle_{g_0}|\mu_{g_0},\\
&+\int_M \frac{1}{N_0}|\langle \pounds_{g_0,conf}Y_n ,\nabla (Y-Y_n)\rangle_{g_0}|\mu_{g_0},\\
&\lesssim ||N^{-1}_0||_{C^0}\Big(||\nabla(Y-Y_n)||_{L^2}||\nabla Y||_{L^2} \\
&+ ||\nabla Y_n||_{L^2}||\nabla (Y-Y_n)||_{L^2}  \Big),
\end{align*}
where the right hand side is well-defined and goes to zero, since $Y\in H_{s+1,-1}$. Which implies that for any $Y\in\mathrm{Ker}(L)$, it holds that
\begin{align*}
\int_M\frac{1}{N_0}\langle \pounds_{g_0,conf}Y ,\nabla Y\rangle_{g_0}\mu_{g_0}=-\lim_{n\mapsto\infty}\int_M\langle \mathrm{div}_{g_0}\left(\frac{1}{N_0}\pounds_{g_0,conf}Y_n \right),Y_n\rangle_{g_0}\mu_{g_0}.
\end{align*}
Again, consider
\begin{align*}
&\Big|\int_M\left( \langle \mathrm{div}_{g_0}\left(\frac{1}{N_0}\pounds_{g_0,conf}Y \right),Y\rangle_{g_0} - \langle \mathrm{div}_{g_0}\left(\frac{1}{N_0}\pounds_{g_0,conf}Y_n \right),Y_n\rangle_{g_0}  \right)\mu_{g_0}\Big|\leq \\
&\int_M|\langle \mathrm{div}_{g_0}\left(\frac{1}{N_0}\pounds_{g_0,conf}(Y-Y_n) \right),Y\rangle_{g_0} |\mu_{g_0} +\int_M|\langle \mathrm{div}_{g_0}\left(\frac{1}{N_0}\pounds_{g_0,conf}Y_n \right),Y-Y_n\rangle_{g_0}|\mu_{g_0}, 
\end{align*}
and notice that $\mathrm{div}_g\left(\frac{1}{N_0}\pounds_{g_0,conf}Y\right)=-\frac{1}{N^2_{0}}\pounds_{g_0,conf}Y(\nabla N_{0},\cdot) + \frac{1}{N_0}\Delta_{g_0,conf}Y$, which shows that
\begin{align*}
|\langle\mathrm{div}_g\left(\frac{1}{N_0}\pounds_{g_0,conf}Y\right),Y\rangle_g|&\lesssim \frac{1}{N^{2}_0}|\nabla Y|_g|DN_0|_g|Y|_g + \frac{1}{N_0}|\Delta_{g_0,conf}Y|_g|Y|_g,\\
&\lesssim \frac{1}{N^{2}_0}|DN_0\otimes \nabla Y|_g|Y|_g + \frac{1}{N_0}|\nabla^2Y|_g|Y|_g.
\end{align*}
Since $DN_0\in H_{s+2,\delta+2}$ and $DY\in H_{s+1,0},$ with $\delta> -\frac{n}{2}$, then the multiplication property gives us that $DN_0\otimes DY\in H_{s+1,1}\subset L^2_{1}$. Thus, since $N^{-1}_0\in C^0$, we get
\begin{align*}
\int_M|\langle \mathrm{div}_{g_0}\left(\frac{1}{N_0}\pounds_{g_0,conf}(Y-Y_n) \right),Y\rangle_{g_0} |\mu_{g_0}&\lesssim ||DN_0\otimes \nabla (Y-Y_n)||_{L^2_1}||Y||_{L^2_{-1}} \\
&+  ||\nabla^2(Y-Y_n)||_{L^2_{1}}||Y||_{L^2_{-1}}\xrightarrow[n\rightarrow\infty]{} 0.
\end{align*}
Similarly, we get that
\begin{align*}
\int_M|\langle \mathrm{div}_{g_0}\left(\frac{1}{N_0}\pounds_{g_0,conf}Y \right),Y-Y_n\rangle_{g_0}|\mu_{g_0}&\lesssim ||DN_0\otimes \nabla Y||_{L^2}||Y-Y_n||_{L^2_{-1}}\\
& + ||\nabla^2Y||_{L^2_{1}}||Y-Y_n||_{L_{-1}}\xrightarrow[n\rightarrow\infty]{} 0,
\end{align*}
which finally shows that for any $Y\in\mathrm{Ker}(L)$, we get
\begin{align*}
\int_M\frac{1}{N_0}\langle \pounds_{g_0,conf}Y ,\nabla Y\rangle_{g_0}\mu_{g_0}=-\int_M\langle \mathrm{div}_{g_0}\left(\frac{1}{N_0}\pounds_{g_0,conf}Y \right),Y\rangle_{g_0}\mu_{g_0}=0.
\end{align*}
Hence, since $\langle \pounds_{g_0,conf}Y ,\nabla Y\rangle_g=\frac{1}{2}\langle\pounds_{g_0,conf}Y ,\pounds_{Y}g_0\rangle_g=\frac{1}{2}\langle\pounds_{g_0,conf}Y ,\pounds_{g_0,conf}Y\rangle_g$, then, we get that if $Y\in\mathrm{Ker}(L)\subset H_{s+2,\delta}$, it holds that
\begin{align}
\int_M\frac{1}{N_0}|\pounds_{g_0,conf}Y|^2_{g_{0}}\mu_{g_{0}}=0,
\end{align}
which implies that $Y$ is a conformal Killing field of $g_0$.
\end{proof}

\bigskip

From the above results, we get that $L$ is semi-Fredholm from the ellipticity property plus Lemma \ref{elliptic}. Furthermore, Using the tools presented in the previous section, it has already been shown, for instance in Theorem 4.6 in \cite{Maxwell2}, that the following theorem holds:

\begin{thm}
Let $(M,g_0)$ be a $H_{s,\rho}$-asymptotically euclidean manifold with $s>\frac{n}{2}$ and $\rho>-\frac{n}{2}$. Then, $\Delta_{g_{0},conf}$ is and isomorphism acting on $H_{s,\delta}$ for any $-\frac{n}{2}<\delta<\frac{n}{2}-2$.
\end{thm}

The above theorem shows that under our hypotheses on $g_0$, \textit{i.e}, $g_0-e\in H_{s+3,\delta+1}$, with $\delta>-\frac{n}{2}$, $g_0$ does not have any conformal Killing fields in $H_{s+2,\delta}$ for any $-\frac{n}{2}<\delta<\frac{n}{2}-2$. Putting together this statement and the above results, we get the following.

\begin{lemma}
Under the same hypotheses of Lemma \ref{injectivity}, plus the constraint $-\frac{n}{2}<\delta<\frac{n}{2}-2$, we get that $L:\mathcal{E}_2\mapsto \mathcal{Z}$ is an isomorphism.
\end{lemma}
\begin{proof}
From the above discussion, we know that $L:H_{s+2,\delta}\mapsto H_{s,\delta+2}$ is injective and semi-Fredholm. Surjectivity can be extracted by analysing its formal adjoint $L^{*}$. In particular, under our constraints on $\delta$, if $L^{*}=L$, then $L$ is surjective (see, for instance, Theorem 3.7 and Coroally 3.8 in Appendix II in \cite{CB2}). Notice that if $X,Y\in C^{\infty}_{0}$, then it holds that
\begin{align*}
\int_M\langle L(X),Y\rangle_{g_0}\mu_{g_0}&=- \frac{1}{2}\int_M\langle \mathrm{div}_{g_0}\left( \frac{1}{N_0}\left(  \pounds_{X}g_0 - \frac{2}{n}  \mathrm{div}_{g_0}X g_0 \right)\right),Y\rangle_{g_0}\mu_{g_0},\\
&= \frac{1}{2}\int_M\frac{1}{N_0}\langle \pounds_{X}g_0 - \frac{2}{n}  \mathrm{div}_{g_0}X g_0 ,\nabla Y\rangle_{g_0}\mu_{g_0},\\
&= \frac{1}{2}\int_M\frac{1}{2}\frac{1}{N_0}\langle \pounds_{X}g_0 - \frac{2}{n}  \mathrm{div}_{g_0}X g_0 ,\pounds_{Y}g_0- \frac{2}{n}  \mathrm{div}_{g_0}Y g_0\rangle_{g_0}\mu_{g_0},\\
&= \frac{1}{2}\int_M \frac{1}{N_0}\langle \nabla X ,\pounds_{Y}g_0- \frac{2}{n}  \mathrm{div}_{g_0}Y g_0\rangle_{g_0}\mu_{g_0},\\
&=- \frac{1}{2}\int_M \langle X ,\mathrm{div}_{g_{0}}\left(\frac{1}{N_0}\left( \pounds_{Y}g_0- \frac{2}{n}  \mathrm{div}_{g_0}Y g_0\right)\right) \rangle_{g_0}\mu_{g_0},\\
&=\int_M \langle X ,L(Y) \rangle_{g_0}\mu_{g_0},
\end{align*}
which shows that $L^{*}=L$, and thus the isomorphism property follows.
\end{proof}
\bigskip

Applying the above lemma to $D_2\Phi_{(\psi_0,\beta_0)}$, followed by the implicit function theorem, we prove

\begin{thm}\label{thmTS1}
Given a solution  $(g_0,K_0)$ of the constraint equations fixed as above, there exists an $\mathcal{E}_1$-neighbourhood of the initial data $\psi_0\in \mathcal{E}_1$, such that the reduced thin-sandwich equations given by {\rm(}\ref{redeq2}{\rm)} are locally well-posed, \textit{i.e}, there is a unique solution $\beta=\beta(\psi)\in H_{s+2,\delta}$, $s>\frac{n}{2}$ and $-\frac{n}{2}<\delta<\frac{n}{2}-2$, for each $\psi$ sufficiently close to $\psi_0$.
\end{thm}

This theorem implies that  for freely chosen initial data in a neighborhood of $\psi_0$, the thin-sandwich problem is well-posed and the constraint equations (\ref{hamiltonian})-(\ref{momentum}) can be solved in terms of the lapse and shift. Note that in this way we can find solutions with small scalar and electromagnetic fields, not constrained to constant mean curvature, although the mean curvature would  be \textit{close} to constant.

\subsection*{Existence of reference solutions}

In this section we establish conditions under which a manifold euclidean at infinity admits a reference solution of the constraint equations of the form proposed in the previous section, that is, a solution of the vacuum constraint equations with a positive cosmological constant, such that $K=\alpha g$. Whenever this is possible, we can guarantee that there is an open subset in the space of solutions of the constraint equations where the thin-sandwich problem is locally well-posed. It is important to stress that, unlike the compact case, we will find obstructions to the existence of such reference solutions. Such obstructions arise since, in this setting, the momentum constraint is automatically satisfied for solutions of the vacuum constraint equations with positive cosmological constant $\Lambda$ and $\alpha^2=\frac{2\Lambda}{n(n-1)}$, and the Hamiltonian constraint becomes
\begin{align}\label{prescurv0}
R_{\bar{g}}=0.
\end{align}

It is a known fact that on AE-manifolds prescribing zero scalar curvature is non-trivial and poses topological obstructions. In order to understand some of the counterexamples, we consider the Yamabe characterization made in \cite{maxwell-dilts}. There, the Yamabe invariant is generalized to asymptotically euclidean (AE, for short) manifolds and (in Theorem 5.1) it is shown that an $W^p_{2,\delta}$-AE manifold $(M,g)$, with $p>\frac{n}{2}$ and $-\frac{n}{p}<\delta<n-\frac{n}{p}-2$, is Yamabe positive if and only if every non-positive scalar curvature function $R\in L^p_{\delta+2}$ can be attained by a metric conformally equivalent to $g$. Then, the authors relate the Yamabe classification for such AE manifolds with the Yamabe classification of their conformal compactifications (see Proposition 5.4 in \cite{maxwell-dilts}). They conclude that such AE manifolds and their conformal compactifications have the same Yamabe type. This yields a wide variety of AE manifolds which do not admit a metric with zero scalar curvature. For instance, consider a $H_{\tilde{s}+2,\delta}$-AE manifold $(M,g)$, with $\tilde{s}>\frac{n}{2}$ and $-\frac{n}{2}<\delta<\frac{n}{2}-2$. We then have the  continuous embedding $H_{\tilde{s},\delta}\subset W^p_{2,\delta_p}$, with $p=\frac{2n}{n-2\tilde{s}+4}$, $\delta_p=\delta+\tilde{s}-2$ and $\tilde{s}<\frac{n}{2}+2$. It is a matter of computations to see that, under these conditions, $p>\frac{n}{2}$ and $-\frac{n}{p}<\delta_p<n-\frac{n}{p}-2$. Thus, we are under the hypotheses of Theorem 5.1 and Proposition 5.4 of \cite{maxwell-dilts}. Then, if $(M,g)$ does not admit some Yamabe positive conformal compactification, then $R_g$ has to be negative somewhere. This is a consequence of the fact that such a $H_{\tilde{s}+2,\delta}$-asymptotically flat metric with non-negative scalar curvature can be conformally transformed to zero scalar curvature, and that a metric with zero scalar curvature is Yamabe positive (see \cite{maxwell-dilts}). Thus, for instance, using the results established in \cite{schoen-yau}, AE manifolds which are obtained by removing a finite number of points from $M^n\#T^n$, with $M$ closed, do not admit zero scalar curvature, and thus, no reference solutions for the constraint equations of the type we are proposing exist. Furthermore, some even simpler counterexamples exist, such as AE manifolds resulting from the removal of points from the torus $T^n$.

Taking into account the above discussion, we see that the most general statement we can make is the following one.

\begin{thm}\label{thmTS2}
On any $n$-dimensional manifold euclidean at infinity, $n\geq 4$, which admits a $H_{s+3,\delta+1}$-Yamabe positive metric, with $s>\frac{n}{2}$ and $-\frac{n}{2}<\delta<\frac{n}{2}-3$, there is an open subset in the space of solutions of the Einstein constraint equations where the thin sandwich problem is locally well-posed.
\end{thm}
 
We now intend to state more explicit conditions which guarantee that an AE manifolds carries a metric with zero scalar curvature. The idea is that the kind of condition we will present will allow us to exhibit relevant examples. Thus, consider the following proposition.

\begin{prop}\label{poincare}
Suppose that $(M^n,g)$ is a $H_{s+1,\delta}$-asymptotically euclidean manifold, with $n\geq 3$, $s> \frac{n}{2}$ and $-1\leq \delta< \frac{n}{2} - 2$. Let $f\in H_{s,\delta+2}$ be a function on $M$ such that its negative part $f_{-}$ satisfies 
\begin{equation}
||f_{-}||_{C^0_{2}}<\frac{1}{C^2_g},
\end{equation} 
where $C_g$ is lowest vale for the Poincar\'e constant {\rm(}which makes the Poincar\'e inequality hold{\rm)}. Then, the equation
\begin{align}
\Delta_gu-fu=0
\end{align}
admits a unique positive solution such that $u-1\in H_{s+2,\delta}$.
\end{prop}

\begin{proof}
We search for solutions of the form $u=1+\phi$, with $\phi \in H_{s+2,\delta}$ satisfying
\begin{align}\label{eq2}
\Delta_g\phi-f\phi=f.
\end{align}
The operator $L=\Delta_g-f:H_{s+2,\delta}\to H_{s,\delta+2}$ is continuous, elliptic and formally self-adjoint \cite{CB1}. Furthermore, if $\phi\in {\rm Ker} L$ then 
\begin{align*}
\int_M\big(|D\phi|^2+f\phi^2\big)\mu_g=0.
\end{align*}
Thus, if the following Brill-Cantor-Maxwell-type condition
\begin{align}\label{BCM}
\int_M\big(|D\phi|^2+f\phi^2\big)\mu_g>0 \quad \mbox{ for all } \quad \phi\in H_{s+2,\delta}, \:\: \phi\not\equiv 0,
\end{align}
holds, then $L$ is an isomorphism between our chosen functional spaces, and (\ref{eq2}) admits a unique solution in $H_{s+2,\delta}$. Now, it is clear that
\begin{align*}
\int_M\big(|D\phi|^2+f\phi^2\big)\mu_g\geq\int_M\big(|D\phi|^2+f_{-}\phi^2\big)\mu_g=\int_M\big(|D\phi|^2-|f_{-}|\phi^2\big)\mu_g.
\end{align*}
However, since $\phi\in H_{s+2,-1}$ by hypotheses, in particular we have that $\phi\in L^2_{-1}$, which is equivalent to $|\phi|^2(1+d^2)^{-1}\in L^1$. Hence,
\begin{align*}
\int_M|f_{-}|\phi^2\mu_g&=\int_M|f_{-}|(1+d^2)|\phi|^2(1+d^2)^{-1}\mu_g\leq \Big(\sup_{x\in M}(1+d^2)|f_{-}|\Big)||\phi||^2_{L^2_{-1}}.
\end{align*}
Our assumptions also imply that $H_{s,\delta+2}\subset C^0_2$. Therefore, $f_- \in C^0_2$ and 
\begin{align*}
\left|\int_Mf_{-}\phi^2\mu_g\right|&\leq ||f_{-}||_{C^{0}_2}||\phi||^2_{L^2_{-1}}\leq C^2_g||f_{-}||_{C^{0}_2}||D\phi||^2_{L^2},
\end{align*}
where  we have used the variant of the Poincar\'e inequality valid in this context (see \cite{maxwell-dilts} for a proof). Using this information in the Brill-Cantor-Maxwell condition, one obtains
\begin{align*}
\int_M\big(|D\phi|^2+f\phi^2\big)\mu_g\geq\big(1-C_g^2||f_{-}||_{C^{0}_2}\big)||D\phi||^2_{L^2}, \quad \mbox{ for all } \quad \phi\in H_{s+2,\delta}.
\end{align*}
Thus, if $||f_{-}||_{C^{0}_2}<1/C^{2}_g$, then (\ref{BCM}) holds, and thus there is unique solution of (\ref{eq2}) in $H_{s+2,\delta}$. The positivity of the solution is a standard consequence of the weak Harnack inequality \cite{Trudinger} (see, for instance, the proof of theorem 11.3, chapter VII, in \cite{CB2}).
\end{proof}

\begin{coro}\label{prescurv1}
Suppose that $(M,e)$ is a manifold euclidean at infinity, which admits a $H_{s+2,\delta}$-asymptotically euclidean metric $g$, with $n\geq 3$, $s> \frac{n}{2}$ and $-1\leq\delta< \frac{n}{2} - 2$, such that $||(R_g)_{-}||_{C^0_{2}}<4\frac{n-1}{n-2}\frac{1}{C^2_g}$, with $C_g$ the best Poincar\'e constant for $g$. Then, there exists a reference solution for the vacuum constraint equations with positive cosmological constant $\Lambda$ on $M$ of the form $(\bar{g},\bar{K}=\alpha \bar{g})$, with $\alpha$ a positive constant and $\bar{g}$ a $H_{s+2,\delta}$-asymptotically euclidean metric.
\end{coro}
\begin{proof} Note that the momentum constraint is automatically satisfied for solutions of the vacuum constraint equations with positive cosmological constant $\Lambda$ and $\bar{K}=\alpha \bar{g}$. Furthermore, picking $\alpha^2=\frac{2\Lambda}{n(n-1)}$, the Hamiltonian constraint becomes
\begin{align}\label{prescurv0}
R_{\bar{g}}=0.
\end{align}
If we look for solutions of the form $\bar{g}=u^{\frac{4}{n-2}}g$, the above equation becomes 
\begin{align}
\Delta_gu-c_nR(g)u=0,
\end{align}
where $c_n=\frac{1}{4}\frac{n-2}{n-1}$. Thus, noticing that under our hypotheses $R_g\in H_{s,\delta+2}$ and using the Proposition \ref{poincare}, one concludes that if 
\[
||(R_g)_{-}||_{C^0_2}<4\frac{n-1}{n-2}\frac{1}{C^2_g},
\]
then the above equation admits a unique positive solution such that $u-1\in H_{s+2,\delta}$, which implies $\bar{g}-e\in H_{s+2,\delta}$. This proves that $(\bar{g},\bar{K})$ gives a solution of the constraint equations with positive cosmological constant $\Lambda$ and a $H_{s+2,\delta}$-asymptotically flat metric $\bar{g}$. This finishes the proof.   
\end{proof}

\vspace{3mm}


\begin{remark}
The above result establishes a sufficient condition for the existence of umbilical reference solutions of the constraint equations. It is interesting to stress that the condition $||R_{g_{-}}||_{C^0_{2}}<4\frac{n-1}{n-2}\frac{1}{C^2_g}$ is \textit{open}. This allows us to produce some interesting examples of reference solutions for the thin-sandwich problem, as follows.
\end{remark}

\begin{exmp}
On $(\mathbb{R}^n,e)$, with $e$ the euclidean metric, there is an explicit estimate for the Poincar\'e constant $C_e$ \cite{bartnik2}. The above remark shows that if we take perturbations of the form $g=(1+f)e$, with $f\in H_{s+3,\delta}$, $s>\frac{n}{2}$, $-1\leq\delta<\frac{n}{2}-2$, being small enough, then we can guarantee that $||R_{g_{-}}||_{C^0_{2}}<4\frac{n-1}{n-2}\frac{1}{C^2_g}$. Thus, $g$ admits a conformal deformation to zero scalar curvature. The same could be said for $g'=g+h$, $h\in C^{\infty}_0$ sufficiently small in the same $H_{s+3,\delta}$-norm. It is worth noticing that metrics of the form $g=(1+f)e$ provide the Newtonian approximation of an initial data set for the constraint equations, which is a relevant situation in physics. Furthermore, we are guaranteeing existence of solutions for small compactly supported perturbations of such metrics. 
\end{exmp}

\begin{exmp}
Consider the Schwarzschild manifold, described by $(\mathbb{R}^n\backslash\{0\},g_{{\rm sc}})$, where 
\begin{align*}
g_{{\rm sc}}=\left(1+\frac{m}{2|x|^{n-2}}\right)^{\frac{4}{n-2}}\delta.
\end{align*} 
It is well-known that $R_{g_{{\rm sc}}}=0$ and it is straightforward to see that $|g_{{\rm sc}}-\delta|^2_{\delta}=O(|x|^{-2(n-2)})$, and thus $(1+|x|^2)^{\frac{\delta}{2}}|g_{{\rm sc}}-\delta|_{\delta}\in L^2$ iff $\delta<\frac{n}{2}-2$. A similar reasoning shows that $g_{{\rm sc}}$ is $H_{s,\delta}$-asymptotically flat for any $s>0$ and every $\delta<\frac{n}{2}-2$. Thus, $g_{{\rm sc}}$ provides an AE vacuum solution of the constraint equations, trivially satisfying $||R_{g_{sc}}||_{C^0_2}<4\frac{n-1}{n-2}\frac{1}{C^2_{g_{sc}}}$. Thus, similarly to the above example, we can consider $H_{s+3,\delta}$-small perturbations of the Schwarzschild metric, which can be deformed to zero scalar curvature and thus used a part of reference solutions for the thin-sandwich problem. 
\end{exmp}

\section{Linearisation at symmetric points - Compact case}

The aim of this section is to study the behavior of the non linear operator $\Phi$ around solutions where its linearization $D_2\Phi$ is not an isomorphism. We still consider reference solutions of the form treated above, that is, pairs $(\psi_0,\beta_0)$ induced from a solution of the form $(g_0, K=\alpha g_0)$, with $S_0=0$. There will be no need of assumptions on the sources generating $\epsilon$, besides their regularity. For ease of exposition, we will suppose $M$ is a closed $n$-dimensional smooth manifold, $n\geq 3$.

In \cite{ADRL} it has been shown that for a solution of the constraint equations of the type proposed above,  $D_2\Phi$ is an isomorphism if $g_0$ does not have any conformal Killing fields. There, it is shown that we can always find a smooth solution $g_0$ without conformal Killing fields. Here, we intend to focus on the case when $g_0$ admits Killing fields. It has already been pointed out that, for such metrics, uniqueness of solutions of the reduced thin sandwich equations must fail \cite{B-O},\cite{giulini1},\cite{bartnik}. Thus, our expectation is that, at least in a range of situations, existence does not fail, and that we can even bifurcate solutions at such \textit{singular} metrics.

We will begin by showing that, under suitable conditions, we can remove the space of Killing fields from the picture and get a well-posed problem on the complementary spaces. Being more precise, this time we consider $\Phi:U\subset \mathcal{E}_1 \times \mathcal{E}_2\mapsto \mathcal{Z}$, where now the Sobolev spaces involved in the definitions of $\mathcal{E}_1, \mathcal{E}_2$ and $\mathcal{Z}$ are not weighted. Here, $U$ is a neighborhood of a smooth solution $(\psi_0,\beta_0)$ constructed from $(g_0,K=\alpha g_0)$. Following the same procedure outlined in \cite{ADRL}, it is clear that such a solution always exists, since the constraints would reduce to the following equation
\begin{align}\label{kw}
R_{g_0}=2\Lambda+2\epsilon-\alpha^2n(n-1).
\end{align}
Thus, supposing that the energy density $\epsilon$ is smooth and using the fact that $M$ is compact, one can pick $\alpha$ sufficiently large so that the right-hand side of the previous equation is negative somewhere. Hence, Kazdan-Warner theorem guarantees the existence of a smooth solution \cite{KW}.

Under the above conditions, the linearization $L\doteq D_2\Phi_{(\psi_0,\beta_0)}:\mathcal{E}_2\to \mathcal{Z}$ is elliptic and formally self adjoint, and ${\rm Ker} L$ is the space of conformal killing vector fields of $g_0$. In particular, we have
\begin{align*}
H_{s+2}&={\rm Ker} L\oplus {\rm Ker} L^{\perp},\\
H_{s}&={\rm Ker}L\oplus {\rm Im}L,
\end{align*}
and thus $L:{\rm Ker}L^{\perp}\mapsto {\rm Im}L $ is an isomorphism and both ${\rm Ker}L^{\perp}$ and ${\rm Im}L$ are closed subsets, and thus Banach. Thus, we can prove the following lemma.

\begin{lemma}\label{lemmasharp}
Let $M$ be an $n$-dimensional compact smooth manifold. Let $g_0\in H_{s+3}$ be a Riemannian metric and consider the following map
\begin{align}
\begin{split}
& \widetilde{\Phi}_{g_0}: U\subset H_{s+1} \times H_{s+1}\times {\rm Im} L\times {\rm Ker} L^{\perp}\to H_{s}\\
& (\dot{g},\epsilon,S,\beta)\mapsto \mathrm{div}_{g_0} \Bigg(\sqrt{\frac{2\epsilon_{\Lambda}-R_{g_0}}{(\mathrm{tr}_{g_0}\gamma)^{2}-|\gamma|^{2}_{g_0}}}\big(\gamma-\mathrm{tr}_{g_0} \gamma\,
g_0\big)\Bigg) - S,
\end{split}
\end{align}
where $U$ is a neighbourhood of $(\psi_0=(\dot{g}_0,\epsilon_0,0),\beta_0)$ and $s>\frac{n}{2}$. If the space of conformal Killing vector fields of $g_0$ consists merely of Killing vector fields, then there are open subsets $U_1\subset \widetilde{\mathcal{E}}_1\doteq H_{s+1}\times H_{s+1}\times {\rm Im }L^{\perp}$ and $U_2\subset {\rm Ker}L^{\perp}$, with $\psi_0\in U_1$ and $\beta_0\in U_2$ such that the equation
\begin{align}
\widetilde{\Phi}_{g_0}(\psi,\beta)=0
\end{align}
has a unique solution $\beta=\beta(\psi)\in U_2$ for all $\psi\in U_1$.
\end{lemma}
\begin{proof}
By definition, we have $\widetilde{\Phi}_{g_0}:U\subset \widetilde{\mathcal{E}}_1\times {\rm Ker}L^{\perp}\to H_s$, $s>\frac{n}{2}$. Thus, if ${\rm Im}(\widetilde{\Phi}_{g_0})\subset {\rm Im}L$, then, because of the above arguments, $D_2{\widetilde{\Phi}_{g_0}}{}_{(\psi_0,\beta_0)}:{\rm Ker}L^{\perp}\to {\rm Im}L$ is an isomorphism, and thus the implicit function theorem finishes the proof. Thus, we need to show that ${\rm Im}(\widetilde{\Phi}_{g_0})\subset {\rm Im}L$. Fixing $u\in {\rm Ker} L$ and $(\psi,\beta)\in U$ and denoting 
\[
N^{-1}\doteq \sqrt{\frac{2\epsilon_{\Lambda}-R_{g_0}}{(\mathrm{tr}_{g_0}\gamma)^{2}-|\gamma|^{2}_{g_0}}},
\]
one computes
\begin{align*}
\int_M \langle u, \widetilde{\Phi}_{g_0}(\psi,\beta)\rangle_{g_0} \mu_{g_0} &= \int_{M} \Big\langle u, \mathrm{div}_{g_0} \Big(\frac{1}{N}\big(\gamma-\mathrm{tr}_{g_0} \gamma  {g_0}\big)\Big) - S\Big\rangle_{g_0}\mu_{g_0}\\
&=-\int_M\frac{1}{N}\langle \nabla u, \gamma-\mathrm{tr}_{g_0} \gamma  {g_0}\rangle_{g_0} \mu_{g_0}.
\end{align*}
Since $\gamma - g_0\:\mathrm{tr}_{g_0}\gamma$ is symmetric we have
\[
 \langle \nabla u, \gamma-g_0\mathrm{tr}_{g_0} \gamma \rangle_{g_0}=\frac{1}{2}\langle \pounds_{u}g_0,\gamma- g_0 \mathrm{tr}_{g_0}\gamma \rangle=0
\]
where we used the assumption that $u$ is a Killing vector field. Thus ${\rm Im}(\widetilde{\Phi}_{g_0})\subset {\rm Ker} L^{\perp}$, and hence ${\rm Im}(\widetilde{\Phi}_{g_0})\subset {\rm Im} L$. This proves the claim. 
\end{proof}

\vspace{3mm}

The above lemma shows that, by restricting the functional spaces appropriately, we can still solve the RTSE in a neighborhood of our reference solution, even if it possesses continuous symmetries. All that is required is that every conformal Killing field must be in fact a Killing field. Thus, in such situations, the thin-sandwich problem can be formulated in a way such that it is well-posed around these symmetric reference solutions. 
\begin{exmp}
It might be instructive to present a simple example of a reference solution of the constraint equations satisfying the requirements of the above lemma. Thus, what we have to find is a solution for {\rm(}\ref{kw}{\rm)} for which the space of conformal Killing vectors coincides with the space of Killing vectors. Consider that $M^n=\mathbb{S}^1\times \Sigma^{n-1}$, where $\Sigma^{n-1}$ is a closed $(n-1)$-dimensional manifold. Consider the following metric on $M$
\begin{align}
g=\pi^{*}h_1+\sigma^{*}h_2,
\end{align}
where $\pi$ and $\sigma$ are the projections of $M$ into its first and second factors respectively, and $h_1$ is the standard metric on $\mathbb{S}^1$, which is obviously flat. Thus, we get that $R_g=R_{h_2}$, and {\rm(}\ref{kw}{\rm)} reads
\begin{align*}
R_{h_2}=2\epsilon_{\Lambda}-\alpha^2n(n-1).
\end{align*}  
Clearly, we must demand $2(\epsilon_\Lambda-\Lambda)$ to be a function only on $\Sigma$. If we consider $n-1\geq 3$ and $\alpha^2>\frac{2}{n(n-1)}\min_{x\in \Sigma}\epsilon_{\Lambda}(x)$, then following the same procedure described in \cite{ADRL}, we can guarantee the existence of a smooth solution to the previous equation without conformal Killing fields. This is achieved by using Kazdan-Warner's theorem \cite{KW} combined with a result shown by J. Lohkamp in \cite{Lohkamp}, which guarantees the existence of a smooth metric with scalar curvature equal to $-1$ which has negative definite Ricci tensor. This procedure yields a smooth solution $h_2$ which is conformal to a metric without conformal Killing fields, and this guarantees that $h_2$ does not possess any conformal Killing fields either. In this way, we get that the space of Killing vectors of the metric $g$ is generated by the lift from $\mathbb{S}^1$ to $M$ of any Killing vector field of $h_1$ on $\mathbb{S}^1$. Furthermore, such Killing vectors are the only conformal Killing fields, and thus, this is an example of a setting where the results presented in \cite{ADRL} do not apply, but the above lemma does. In particular, this kind of examples might even be of interest in physics, since they would represent the kind of initial data sets appropriate for Kaluza-Klein theories, which were the starting point for modern extra dimensional theories, such as string theory.   
\end{exmp}
\bigskip
Next, we would like to obtain a qualitative description of the set of solutions of RTSE which do possess conformal Killing vectors. We would expect that around such data the non-uniqueness issues remain, but that we can still guarantee existence, at least under some conditions. In order to address this issue, we will employ Crandall-Robinowitz bifurcation criteria \cite{crandal}. The result we are interested in is the following one.

\begin{thm}[Crandall-Rabinowitz]
Let $X$ and $Y$ be Banach spaces and let $F$ be a $C^1$ mapping of an open neighborhood of a given point $(\bar{\lambda}, \bar{x})\in \mathbb{R}\times X$ into $Y$. Let ${\rm Ker}(D_2F_{(\bar{\lambda}, \bar{x})})= {\rm span}\{{x_0}\}$ be one dimensional and ${\rm codim} ({\rm Im}(D_2F_{(\bar{\lambda}, \bar{x})}))= 1$. Let $D_1F_{(\bar{\lambda}, \bar{x})}\not\in {\rm Im}(D_{2}F_{(\bar{\lambda}, \bar{x})})$. If $Z$ is a complement of ${\rm span}\{{x_0}\}$ in $X$, then the solutions of $F(\lambda, x)=F(\bar{\lambda}, \bar{x})$ near $(\bar{\lambda}, \bar{x})$ form a curve $(\lambda (s), x (s)) = (\bar{\lambda} + \tau (s), \bar{x} + s x_0 + z (s))$, where $s \mapsto (\tau (s), z (s)) \in \mathbb{R} \times Z$ is a $C^1$ function near $s = 0$ and $\tau (0) = \tau' (0) = 0, z (0) = z' (0) = 0$. Moreover, if $F$ is k-times continuously differentiable, so are $\tau(s)$ and $z(s)$.
\end{thm}

It is not a completely trivial task to produce solutions of the constraint equations satisfying all the properties of the above theorem which might be \textit{physically meaningful}.  The examples we will consider in the sequel permit a straightforward application of the theorem. However,  
their physical meaning is not indisputable. 
Given a compact Riemannian manifold $(\Sigma^{n-1}, h)$ with $n\ge 4$, we define a metric of the form $g_0=d\theta\otimes d\theta + f^2(\theta)h$ in  $M^n=\mathbb{S}^1\times \Sigma^{n-1}$, where $\theta$ is the standard angular coordinate on the circle and $f$ is a prescribed, non-constant, strictly positive function. Again, we will consider $K_0=\alpha g_0$ and $S_0=0$. Thus, we need to solve the hamiltonian constraint written in this setting as
\begin{align}
R_{g_0}=\frac{1}{f^2}R_h-2(n-1)\frac{f''}{f}-(n-1)(n-2)=2\Lambda+2\widetilde{\epsilon}_0-\alpha^2n(n-1),
\end{align}
where $\tilde{\epsilon}_0$ represents the energy density of some source. In order to produce a solution of the above equation such that $g_0$ only possess one conformal Killing field, which is not a Killing vector, we can choose $h$ such that $R_h=-1$ and such that $h$ does not have any conformal Killing vectors. Then, we consider the energy density given by
\begin{align*}
2\widetilde{\epsilon}_0(\Lambda,\theta)=-2\Lambda-(n-1)(n(1-\alpha^2)-2)-\frac{1+2(n-1)ff''}{f^2}.
\end{align*}
With this induced energy density, we can produce reference solutions for the RTSE by choosing a smooth lapse $N_0>0$ and shift $\beta_0$ and defining $\dot{g}_0\doteq -2\alpha N_0 g_0 + \pounds_{\beta_0}g_0$. Then, $\Phi(g_0,\dot{g}_0,\widetilde{\epsilon}_0(\Lambda),0,\beta_0)=0$. The theorem presented in \cite{ADRL} is not appropriate in this context since ${\rm Ker} (D_{2}\Phi_{(\psi_0(\Lambda),\beta_0)})$ is the space of conformal Killing fields of $g_0$, and $f(\theta)\partial_{\theta}$ is a conformal Killing field of $g_0$. Therefore, $D_{2}\Phi_{(\psi_0(\Lambda),\beta_0)}$ is not an isomorphism and we cannot apply the implicit function argument. Even the refinement of this theorem, presented in Lemma \ref{lemmasharp}, does not work either, since we have a conformal Killing field, which is not a Killing field. We will now show that we can bifurcate this curve of solutions, parameterized by the cosmological constant $\Lambda$, and get (at least) another family of solutions passing through $(\psi_0(\Lambda),\beta_0)$. With this in mind, first define $\bar{\psi}(\Lambda)\doteq (g_0,\dot{g}_0,\epsilon_0(\Lambda),(\Lambda-\Lambda_0)S)$, where $S\in H_{s}$ is chosen such that is has non-zero projection onto ${\rm span}\{f\partial_\theta\}$. Clearly, $\bar{\psi}(\Lambda)$ and $\psi_0(\Lambda)$ intersect at $\Lambda_0$. Now, define
\begin{align}
\begin{split}
& F:U\subset \mathbb{R}\times H_{s+2} \to H_{s}\\
& (\Lambda,\beta)\mapsto F(\Lambda,\beta)\doteq \Phi(\bar{\psi}(\Lambda),\beta),
\end{split}
\end{align}
where $U$ is an open set chosen such that $\Phi(\bar{\psi}(\Lambda),\beta)$ is well-defined. Thus, 
\[
{\rm Ker}(D_2F_{(\Lambda_0,\beta_0)})={\rm span}\{f\partial_\theta\}
\]
and, since $D_2F_{(\psi(\Lambda_0),\beta_0)}$ is an elliptic formally self-adjoint operator, we have
\[
H_s={\rm Ker}(D_2F_{(\Lambda_0,\beta_0)})\oplus {\rm Im}(D_2F_{(\Lambda_0,\beta_0)}).
\]
Thus ${\rm codim} ({\rm Im}(D_2F_{(\Lambda_0,\beta_0)}))=1$. Furthermore, a straightforward computation gives
\begin{align*}
D_1F_{(\Lambda_0,\beta_0)}=-S\not\in {\rm Im}(D_2F_{(\Lambda_0,\beta_0)}),
\end{align*}
since $S$ has non-zero projection onto ${\rm Ker}(D_2F_{(\Lambda_0,\beta_0)})$. Thus, Crandall-Rabinowitz's theorem applies, and we get that the solutions of $F(\Lambda,\beta)=0$  around $(\Lambda_0,\beta_0)$ form a curve $(\Lambda(s),\beta(s))=(\Lambda_0+\tau(s),\beta_0+sf\partial_\theta+z(s))$ with
\[
s\mapsto (\tau(s),z(s))\in \mathbb{R}\times {\rm Im}(D_2F(\Lambda_0,\beta_0))
\]
and $\tau(0)=\tau'(0)=0$, $z(0)=z'(0)=0$. 

\subsection*{Neighbourhoods of umbilical reference solutions with conformal Killing fields}

From the above discussion we know that the results presented in \cite{ADRL} can be sharpened in the following sense. Given an umbilical reference solution of the constraint equations of the form $(g_0,K_0=\alpha g_0)$, where the momentum density is zero, and we may have a non-zero energy density $\epsilon_0$, suppose that the space of conformal Killing fields of $g_0$ consists merely of Killing fields. We have seen that in this case our functional spaces can be refined so that the implicit function argument can be applied. 
Thus, what can be said about  umbilical reference solutions which admit conformal Killing fields such that the implicit function argument cannot be applied. A natural question at this point would be whether in these cases we can find another umbilical reference solution \textit{close} to the original one, which induces a reference solution $(\psi,\beta)$ for the RTSE for which the implicit function argument \textit{can} be applied. 

Thus, we consider a smooth solution of the vacuum constraint equations of the form $(g_0,K=\frac{\tau}{n} g_0)$, where $\tau$ is constant which represents the mean curvature of the embedded hypersurface $M\hookrightarrow M\times\mathbb{R}$. We suppose that $g_0$ has non-trivial conformal Killing fields. Since we are considering vacuum, $g_0$ satisfies 
\begin{align*}
R_{g_0}=2\Lambda-\frac{\tau^2}{n}(n-1)
\end{align*}
and, if $\Lambda\neq 0$, we suppose that $2\Lambda-\frac{\tau^2}{n}(n-1)<0$ which means that $g_0$ has negative constant scalar curvature. Our aim is to find another solution of the form $(\bar{g},\frac{\tau}{n}\bar{g})$, with $R_{\bar{g}}=R_{g_0}$ such that $\bar{g}$ is close to $g_0$ and ${\rm Ker}(\Delta_{\bar{g}, {\rm conf}})=\{0\}$, so that $(\bar{g},\frac{\tau}{n}\bar{g})$ induces a reference solution $(\bar{\psi}_0,\bar{\beta}_0)$ where the implicit function argument can be applied. In order to do this, we will need some auxiliary results. 

First, consider the set of smooth Riemannian metrics on a closed manifold $M^n$, with $n\geq 3$, denoted by $\mathcal{M}$, endowed with the distance function introduced in \cite{Eigenlap} and defined by
\begin{align}
d(g_1,g_2)\doteq d_0(g_1,g_2) + d_1(g_1,g_2),
\end{align}
where $d_0$ is the distance defined on the space of symmetric $(0,2)$-tensor fields on $M$, denoted by $\mathcal{S}(M)$ given by
\begin{align*}
d_0(g_1,g_2)=\sum_{k=0}^{\infty}\frac{1}{2^k}\frac{p_k(g_1-g_2)}{1+p_k(g_1-g_2)}
\end{align*}
where the semi-norms $p_k$ are defined by
\begin{align*}
p_k(h)=\sum_{\ell=0}^k\sup_{x\in M}|\nabla^{\ell}h|_e,
\end{align*}
for some fixed  smooth Riemannian metric $e$ on $M$ and the corresponding Riemannian connection $\nabla$.  The distance $d_1$ is defined on $\mathcal{M}$ as follows:
\begin{align}
\begin{split}
d_1(g_1,g_2)&=\sup_{x\in M}{d_1}_{x}({g_1}_{x},{g_2}_{x}),\\
{d_1}_{x}({g_1}_{x},{g_2}_{x})&=\inf\{\delta>0 : e^{-\delta}{g_1}_{x}<{g_2}_x< e^{\delta}{g_{1}}_x\},
\end{split}
\end{align}
where $g_1<g_2$ means $g_2-g_1\in \mathcal{S}_x(M)$ is positive definite. In \cite{Eigenlap}, the authors show that $(\mathcal{M},d)$ is a Fr\'echet space. We will consider this topology on $\mathcal{M}$ from now on. In this context, in \cite{genericmetrics} it is shown that the subset of $\mathcal{M}$ with no conformal symmetries is dense in $\mathcal{M}$. The key observation to be made here is that the density of such subset in $\mathcal{M}$ implies that, given any $C^k$ neighborhood $U$ of an element $g_0\in \mathcal{M}$ (with respect to the semi-norm $p_k$), we can always find a metric $g\in \mathcal{M}$ which is sufficiently close to $g_0$ in the distance $d$, so that $g\in U$ and $g$ does not possess any conformal Killing fields.

\begin{prop}
Consider a compact $n$-dimensional manifold $M$, with $n\geq 3$. Given a smooth solution of the vacuum constraint equations on $M$ of the form $(g_0,K_0=\frac{\tau}{n} g_0)$ with constant mean curvature $\tau$ and constant scalar curvature $R_{g_0}=-1$, we can find another smooth solution of the vacuum constraint equations with the same mean and scalar curvatures, which is as close to $(g_0,K_0)$ as we want and does not admit any conformal Killing fields.   
\end{prop}

\begin{proof}
Consider $g\in \mathcal{M}$ and $\bar{g}=u^{\frac{4}{n-2}}g$ a conformal rescaling. Thus, 
\begin{align*}
R_{\bar{g}}=u^{-\frac{n+2}{n-2}}\left( uR_g - \frac{4(n-1)}{n-2}\Delta_gu \right).
\end{align*}
Fix $s>\frac{n}{2}+2$ and define the map
\begin{align*}
& F:U\subset H_{s}\times H_s\to H_{s-2},\\
& (g,u)\mapsto \Delta_gu-c_n(R(g)u+u^N),
\end{align*}
where $c_n=\frac{n-2}{4(n-1)}$, $N=\frac{n+2}{n-2}$ and $U$ is a neighbourhood of $(g_0,1)\in H_s\times H_s\subset C^2\times C^2$. It is clear that $F(g_0,1)=0$ and  we also have that
\begin{align*}
D_2F_{(g_0,1)}\cdot v= \Delta_{g_0}v-c_n(N-1)v.
\end{align*}
Since $c_n(N-1)>0$, it follows that $D_2F_{(g_0,1)}:H_s\mapsto H_{s-2}$ is an isomorphism. Thus, applying the implicit function theorem, one concludes that there exist neighbourhoods $U_1\subset H_s$ and $U_2\subset H_s$ of $g_0$ and $1$ respectively, and a unique map $f:U_1\to U_2$ such that $F(g,f(g))=0$, for all $g\in U_1$. Furthermore, since the coefficients are smooth, we can increase the regularity of the solution $u=f(g)$ and get a smooth solution. Moreover, taking a small enough neighborhood of $(g_0,1)$ we can guarantee that $u>0$. Thus, we get $\bar{g}\doteq u^{\frac{4}{n-2}}g\in\mathcal{M}$ satisfying $R_{\bar{g}}=-1$. Also, given any $\epsilon>0$, we can fit the above procedure so that 
\[
||\bar{g}-g||_{C^k}\leq C||\bar{g}-g||_{H_{\frac{n}{2}+k}}<\epsilon
\]
for any integer $k>0$, which also shows that, setting $\bar{K}=\frac{\tau}{n}\bar{g}$, then $||\bar{K}-\frac{\tau}{n}g||_{C^k}<\epsilon'$. 


The final argument is given by noticing that, since the subset of elements in $\mathcal{M}$ without any conformal Killing fields in dense in $\mathcal{M}$, given a $C^k$-neighbourhood of $g_0$, we can always find an element $g\in \mathcal{M}$ such that $||g-g_0||_{C^k}$ is as small as we want, and such that $g$ does not admit conformal Killing fields. Then, choosing such $g\in U_1$, we get $u=f(g)$ from the implicit function argument described above, and a smooth metric $\bar{g}=u^{\frac{4}{n-2}}g$, which cannot admit any conformal Killing fields, since that would contradict the fact that $g$ does not. Taking into account all these considerations, it is clear that $(\bar{g},\bar{K})$ solve the same vacuum constraints as $(g_0,K_0)$ and the two solutions have the same mean and scalar curvatures. Moreover, we can construct $(\bar{g},\bar{K})$ so that it is as close to $(g_0,K_0)$ as we want in any $C^k$-topology. This concludes the proof.
\end{proof}

\vspace{3mm}

It is clear that the a mere rescaling argument permits to extend  the condition $R_{g_0}=-1$ to the more general $R_{g_0}=2\Lambda-\frac{\tau^2}{n}(n-1)<0$. Hence, the same procedure described above can be applied to the general situation and  the following theorem holds.

\begin{thm}\label{generic-umbilical}
Given any umbilical smooth solution to the vacuum constraint equations $(g_0,K_0)$ on a compact $n$-dimensional manifold $M$ satisfying $2\Lambda-\frac{\tau^2}{n}(n-1)<0$, with $n\geq 3$, there is another smooth solution $(\bar{g},\bar{K})$, which is as close as we want to $(g_0,K_0)$ in any $C^k$-topology and has the same mean and scalar curvatures as $(g_0,K_0)$, for which the induced solution $(\bar{\psi}_0,\bar{\beta}_0)$ for the RTSE  admits a neighborhood where the RTSE are well-posed.  
\end{thm}

The above theorem, for instance, shows that any umbilical solution of the vacuum (without cosmological constant) constraint equations, either produces reference solutions of the RTSE such that these equations are well-posed in a neighborhood of this data, or there is a another solution close to it, such that the previous claim holds. 

\section*{Acknowledgements}

We would like to thank professor Justin Corvino for reading a previous version of this paper and making several valuable comments and suggestions.


\end{document}